%% file: 202205_ICASSP_robin_sdr_medium_rare.tex
\documentclass{article}
\usepackage{spconf}
\usepackage{flushend}  
\usepackage{blindtext, graphicx}
\usepackage{bm}
\usepackage{booktabs}
\usepackage{subfig}
\usepackage[font=footnotesize]{caption}
\usepackage{hyperref}
\usepackage{xspace}
\usepackage[binary-units]{siunitx}
\usepackage{nicefrac}

\usepackage{tikz}
\usepackage{tikz-3dplot}
\usetikzlibrary{arrows.meta}
\usetikzlibrary{angles}
\usetikzlibrary{calc,positioning,intersections}

\usepackage[style=ieee]{biblatex}
\usepackage{localmath}

\newcommand{\SISDR}{\ensuremath{\operatorname{SI-SDR}}\xspace}

\newcommand{\SDR}{\ensuremath{\operatorname{SDR}}\xspace}

\newcommand{\SIR}{\ensuremath{\operatorname{SIR}}\xspace}

\newcommand{\SAR}{\ensuremath{\operatorname{SAR}}\xspace}
\newcommand{\mireval}{\texttt{mir\_eval}\xspace}
\newcommand{\cisdr}{\texttt{ci\_sdr}\xspace}
\newcommand{\sigsep}{\texttt{sigsep}\xspace}
\newcommand{\solve}{\texttt{solve}\xspace}
\newcommand{\cgd}{\texttt{CGD10}\xspace}
\newcommand{\fpsingle}{{fp32}\xspace}
\newcommand{\fpdouble}{{fp64}\xspace}
\newcommand{\bsseval}{bss\_eval\xspace}

\newcommand{\half}{\nicefrac{1}{2}}

\providecommand{\vshat}{\hat{\vs}}
\providecommand{\starget}{\vs^{\text{target}}}
\providecommand{\einterf}{\ve^{\text{interf}}}
\providecommand{\eartif}{\ve^{\text{artif}}}

\graphicspath{{figures/}}

\begin{document}
\ninept

\title{SDR --- Medium rare with fast computations}

\name{Robin Scheibler%
\thanks{
The research presented in this paper is reproducible. Code and data are available at \protect\url{http://github.com/fakufaku/fast_bss_eval}.%
}}
\address{LINE Corporation}

%

\maketitle

\begin{abstract}
  We revisit the widely used \bsseval metrics for source separation with an eye out for performance.
  We propose a fast algorithm fixing shortcomings of publicly available implementations.
  First, we show that the metrics are fully specified by the squared cosine of just two angles between estimate and reference subspaces.
  Second, large linear systems are involved.
  However, they are structured, and we apply a fast iterative method based on conjugate gradient descent.
  The complexity of this step is thus reduced by a factor quadratic in the distortion filter size used in \bsseval, usually 512.
  In experiments, we assess speed and numerical accuracy.
  Not only is the loss of accuracy due to the approximate solver acceptable for most applications, but the speed-up is up to two orders of magnitude in some, not so extreme, cases.
  We confirm that our implementation can train neural networks, and find that longer distortion filters may be beneficial.
\end{abstract}
\begin{keywords}%
source separation, performance evaluation, \bsseval, signal-to-distortion ratio, conjugate gradient descent
\end{keywords}


\section{Introduction}

Blind source separation (BSS) refers to a family of technique that can be used to recover signals of interests from their mixtures using only minimal prior information.
It has broad applications, but we focus on audio, e.g., for speech~\cite{makinoAudioSourceSeparation2018} and music~\cite{canoMusicalSourceSeparation2019} signals.
BSS comes in many flavors.
Deep neural networks (DNN) have been successfully used to separate multiple speakers from a single microphone's signal~\cite{hersheyDeepClusteringDiscriminative2016}.
Approaches based on independent component analysis work on the determined case where there as many sources as measurements~\cite{comonHandbookBlindSource2010}.
Convolutive mixtures, such as found in audio, may be handled by independent vector analysis (IVA)~\cite{kimIndependentVectorAnalysis2006a,hiroeSolutionPermutationProblem2006}.
Finally, overdetermined IVA tackles the case where redundant measurements are available~\cite{scheiblerIndependentVectorAnalysis2019,boeddeker_convolutive_2021}.
Performance evaluation is key to developing new algorithms and requires relevant metrics to be available.
For audio BSS, the \textit{\bsseval} metrics, i.e., signal-to-distortion, interference, and artifact ratios (SDR, SIR, and SAR, respectively), are a de facto standard~\cite{vincentPerformanceMeasurementBlind2006} and have been routinely used for the evaluation of new algorithms, e.g.~\cite{onoStableFastUpdate2011,hersheyDeepClusteringDiscriminative2016,scheiblerIndependentVectorAnalysis2019,stoterOpenUnmixReferenceImplementation2019}.
They decompose the estimated signals into orthogonal components corresponding to target sources and noise, as illustrated in \ffref{bsseval_illustration}.
Some amount of distortion in the estimate is allowed by forgiving a 512 taps filter.

It has been argued that these filters may actually be detrimental, especially for mask-based approaches~\cite{LeRoux:2018tq}.
As a countermeasure, the scale-invariant \SDR (\SISDR), i.e., the \SDR with a single tap filter, has been proposed~\cite{LeRoux:2018tq}.
Subsequently, it has been used for end-to-end training of separation networks~\cite{luoTaSNetTimeDomainAudio2018,scheiblerSurrogateSourceModel2021}.
However, the classic \bsseval \SDR has been recently vindicated and shown to outperform the \SISDR as a loss for end-to-end training of linear separation systems~\cite{boeddeker_convolutive_2021}.
Nevertheless, some computational challenges remain.
Computation of the optimal filters involves inversion of very large matrices, with cubic complexity using direct solvers.
This may be crippling if many short signals have to be evaluated, e.g. in utterance-level permutation invariant training (uPIT)~\cite{kolbaekMultitalkerSpeechSeparation2017}.
Furthermore, publicly available implementations, such as in \mireval, are not as performant as one would desire, leading to long computation times when applied to large datasets, especially when the number of channels increases.
For iterative methods~\cite{onoStableFastUpdate2011,scheiblerIndependentVectorAnalysis2019}, \bsseval metrics have to be evaluated at multiple iterations to assess convergence.

In this work, we propose a highly efficient algorithm for the computation of the \bsseval metrics, i.e. \SDR, \SIR, and \SAR.
First, we provide a finer analysis of the definition of the metrics, leading to computational savings.
These are substantial as they reduce to a minimum steps dealing with the full length of the signals.
Second, to reduce computations for the distortion filters, we propose to use conjugate gradient descent (CGD)~\cite{chanConjugateGradientMethods1996}.
We implement the proposed algorithm in \texttt{pytorch}~\cite{paszkePyTorchImperativeStyle2019}, making it fully differentiable and GPU-enabled\footnote{At the tip of your fingers: \texttt{pip install fast\_bss\_eval}}.
We analyze the trade-off between numerical accuracy and speed in experiments on speech signals.
Our proposed implementation is orders of magnitude faster than publicly available ones.
The simplified steps provide savings for longer signals, while the CGD kicks in for more channels, or longer filters.
For \SDR only computations, we show up to $27\times$ speed-up for 8 channels and a 1024 taps filter compared to~\cite{boeddeker_convolutive_2021}.
We demonstrate successful training of a neural network for source separation using our implementation of the \SDR as the loss.
Interestingly, we find that doubling the length of the filters (\num{1024} taps) leads to further improvements.

\section{Background}

Vectors and matrices are represented by bold lower and upper case letters, respectively.
The norm of vector $\vx$ is written $\| \vx \| = (\vx^\top \vx)^{\half}$.
The convolution of vectors $\vx$ and $\vh$ is denoted $\vx \star \vh$.

We consider the case where we have $M$ estimated signals $\vshat_m$.
Each contains a convolutional mixture of $K$ reference signals $\vs_k$ and an additive noise component $\vb_m$,
\begin{align}
  \vshat_m = \sum\nolimits_{k}\vh_{mk} \star \vs_k + \vb_m, \quad m=1,\ldots,M,
\end{align}
where $\vshat_m$, $\vs_k$, and $\vb_m$ are all real vectors of length $T$.
The length of the impulse responses $\vh_{mk}$ is assumed to be short compared to $T$.
For simplicity, the convolution operation here includes truncation to size $T$.
In most cases, the number of estimates and references is the same, i.e. $M=K$.
We keep them distinct for generality.

\subsection{\bsseval v3.0}

There exists a few variants of the \bsseval metrics~\cite{vincentPerformanceMeasurementBlind2006}, but we concentrate on the so-called v3.0.
It is the most recent and the one implemented in \mireval~\cite{raffelMirEvalTransparent2014} and \cisdr~\cite{boeddeker_convolutive_2021}.
In general, the matching of source estimates to their reference is not known and must be computed.
To this end, the metrics must be computed for each pair $(\vshat_m, \vs_k)$ before the best matching is found.

\bsseval decomposes the estimated signal as shown in \ffref{bsseval_illustration},
\begin{align}
  \starget_{km} = \mP_{k} \vshat_m, \quad
  \einterf_{km} = \mP \vshat_m - \mP_k \vshat_m, \quad
  \eartif_{km} = \vshat_m - \mP \vshat_m.
  \nonumber
\end{align}
Matrices $\mP_k$ and $\mP$ are projection matrices onto the $L$ shifts of $\vs_k$ and of all references, respectively.
Let $\mA_k \in \R^{(T + L - 1) \times L}$ contain the $L$ shifts of $\vs_k$ in its columns and $\mA = [\mA_1,\ldots,\mA_K]$,
then
\begin{align}
  \mP_k = \mA_k(\mA_k^\top \mA_k)^{-1}\mA_k^\top,\quad 
  \mP = \mA(\mA^\top \mA)^{-1}\mA^\top.
  \elabel{proj_matrices}
\end{align}
Then, SDR, SIR, and SAR, in decibels, are defined as follows,
\begin{align}
  \SDR_{km} & = 10\log_{10} \frac{\| \starget_{km} \|^2}{\|\einterf_{km} + \eartif_{km} \|^2}, \elabel{sdr} \\
  \SIR_{km} & = 10\log_{10} \frac{\| \starget_{km} \|^2}{\| \einterf_{km} \|^2}, \elabel{sir} \\
  \SAR_{km} & = 10\log_{10} \frac{\| \starget_{km} + \einterf_{km} \|^2}{\| \eartif_{m} \|^2}. \elabel{sar}
\end{align}
Finally, assuming for simplicity that $K=M$, the permutation of the estimated sources $\pi\,:\,\{1,\ldots,K\} \to \{1,\ldots,K\}$ that maximizes $\sum_k \SIR_{k\,\pi(k)}$ is chosen\footnote{\cisdr~\cite{boeddeker_convolutive_2021} uses the $\SDR$ since it does not compute the $\SIR$.}.

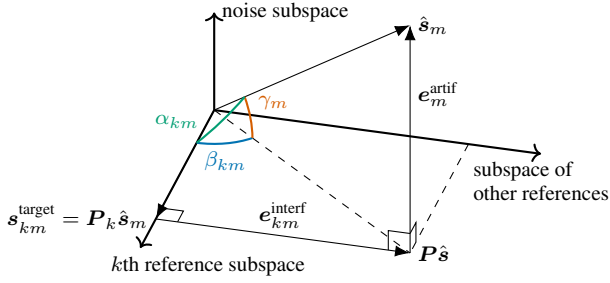
\begin{figure}
  \centering
  \input{project_illustrations.tex}
  \caption{Illustration of the decomposition operated by \bsseval.
    The estimated source $\vshat_m$ is decomposed into $\starget_{km}$, $\einterf_{km}$, and $\eartif_{km}$ by orthogonal projections onto the subspaces spanned by the $L$ shifts of $\vs_k$, of other references, and the noise, respectively.
    In \sref{base_metrics}, we show that \SDR, \SIR, and \SAR are uniquely determined by angles $\alpha_{km}$, $\beta_{km}$, and $\gamma_{km}$, respectively.
  }
  \flabel{bsseval_illustration}
\end{figure}

\subsection{Standard Implementations}

Publicly available implementations of \bsseval~\cite{vincentPerformanceMeasurementBlind2006,raffelMirEvalTransparent2014,boeddeker_convolutive_2021} all follow a fairly straightforward approach for the computations.
They all use a fixed or default value of $L=512$.
We start by defining the autocorrelation matrix of reference $\vs_k$, of all references, and the cross-correlation of $\vs_k$ and estimate $\vshat_m$,
\begin{align}
  \mR_k = \mA_k^\top \mA_k, \quad
  \mR = \mA^\top \mA, \quad
  \vx_{km} = \mA^\top \vshat_m,
  \nonumber
\end{align}
respectively.
Then, \SDR, \SIR, and \SAR are computed as follows.
\begin{enumerate}
  \item Compute $\mR$ and $\vx_{km}$ for all $k$ and $m$.
    The former is a block-Toeplitz matrix containing $\mR_k$ as its diagonal blocks.
    The complexity is $O(K^2 T\log T)$ and $O(KM T\log T)$, respectively, using the fast Fourier transform (FFT)~\cite{cooleyAlgorithmMachineCalculation1965}.
  \item Compute filters $\vh_{km}$ in $O(K L^3 + K M L^2)$ by solving
    \begin{align}
      \mR_k [\vh_{k1},\ldots,\vh_{kM}] = [\vx_{k1},\ldots,\vx_{kM}].
      \elabel{sdr_linear_system}
    \end{align}
  \item Compute filters $\vg_{k\ell}$ in $O(K^3 L^3 + M K^2 L^2)$ by solving
    \begin{align}
      \mR
      \begin{bmatrix} \vg_{11} & \cdots & \vg_{1K} \\ \vdots & \ddots & \vdots \\ \vg_{K1} & \cdots & \vg_{KK} \end{bmatrix}
      & =
      \begin{bmatrix} \vx_{11} & \cdots & \vx_{1M} \\ \vdots & \ddots & \vdots \\ \vx_{K1} & \cdots & \vx_{KM} \end{bmatrix}.
      \elabel{sir_linear_system}
    \end{align}
  \item Compute $\starget_{km} = \vh_{km} \star \vs_k$ in $O(K M T \log L)$.
  \item Compute $\vu_m = \sum_k \vg_{km} \star \vs_k$ in $O(K M T \log L)$.
  \item Compute $\einterf_{km} = \starget_{km} - \vu_m$ and $\eartif_{m} = \vshat_m - \vu_m$.
    Then, compute the $\SDR_{km}$, $\SIR_{km}$, and $\SAR_{km}$ according to \eref{sdr}, \eref{sir}, and \eref{sar}.
    The complexity is $O(KMT)$.
  \item Find the best permutation of the estimated sources in $O(M^3)$ by using the Hungarian algorithm\footnote{While the use of the Hungarian algorithm seems to have been rediscovered recently, it has long been applied in the BSS literature, e.g.,~\cite{drakeSoundSourceSeparation2002,ciaramellaAmplitudePermutationIndeterminacies2003,tichavskyOptimalPairingSignal2004}.}~\cite{kuhnHungarianMethodAssignment1955}.

\end{enumerate}

\section{Proposed Algorithm}

We identified the following inefficiencies in the implementation described above.
Components $\starget$, $\einterf$, and $\eartif$ do not need to be computed.
We show below that we can replace steps 4, 5, and 6 by a simpler computation.
These may not seem very computationally expensive, however, they operate on the full length of the input signals.
Audio signals may be long, e.g., \SI{30}{\second} sampled at \SI{16}{\kilo\hertz} is \num{480000} samples.
The linear systems \eref{sdr_linear_system} and \eref{sir_linear_system} are expensive to solve directly due to $L$ being typically large, e.g., $L=512$.
However, they are Toeplitz and block-Toeplitz, respectively, and efficient algorithms can be applied.
%

\subsection{Efficient Computation using Cosine Metrics}
\seclabel{base_metrics}

Since the \bsseval metrics are not sensitive to the scales of $\vs_k$ and $\vshat_m$, we will assume that all signals are scaled to have unit norm, i.e. $\|\vs_k\| = \|\vshat_m\| = 1$ for all $k,m$.
\begin{definition}[Cosine Metrics]
  We define the following new metrics,
  \begin{align}
    c_{km} & = \vshat_{m}^\top \mP_k \vshat_m = \vx_{km}^\top \mR_k^{-1} \vx_{km} = \vx_{km}^\top \vh_{km}, \\
    d_{m} & = \vshat_{m}^\top \mP \vshat_m = \vz_{m}^\top \mR^{-1} \vz_{m} = \sum\nolimits_k \vx_{km}^\top \vg_{km},
  \end{align}
  where $\vz_m = \mA^\top \vshat_m = [\vx_{1m}^\top,\ldots,\vx_{Km}^\top]^\top$.
\end{definition}

\begin{theorem}
  The \bsseval metrics can be computed as follows,
  \begin{align}
    \SDR_{km} & = f(c_{km}), \elabel{sdr_base} \\
    \SIR_{km} & = f(\nicefrac{c_{km}}{d_m}), \elabel{sir_base} \\
    \SAR_{m} & = f(d_m). \elabel{sar_base} 
  \end{align}
  where $f(x) = 10 \log_{10}\left(\frac{x}{1 - x}\right)$.
  \label{thm:bss_eval_from_base_metrics}
\end{theorem}
\begin{proof}
  The proof follows directly from properties of projection matrices, namely idempotency and self-adjointness.
  Given a real projection operator $\mPi$, these properties mean $\mPi \mPi = \mPi$ and $\mPi = \mPi^\top$, respectively.
  Further, $\mI - \mPi$ is also a projection matrix.
  Thus,
  \begin{align}
    \frac{\|\starget_{km}\|^2}{\|\einterf_{km} + \eartif_{km}\|^2}
    & =\frac{\|\mP_k \vshat_m\|^2}{\|(\mI - \mP_k) \vshat_m \|^2} 
      =\frac{\vshat_m^\top \mP_k \vshat_m}{\vshat_m^\top\vshat_m - \vshat_m^\top \mP_k \vshat_m}, \nonumber
  \end{align}
  and \eref{sdr_base} follows since $c_{km} = \vshat_m^\top \mP_k \vshat_m$ and we assumed $\|\vshat_k\| = 1$.
  From their definition in \eref{proj_matrices}, it is clear that the range space of $\mP$ contains that of $\mP_k$, and thus, $\mP \mP_k = \mP_k \mP = \mP_k$.
  This can be used to show that $\mP - \mP_k$ is also a projection matrix.
  Thus,
  \begin{align}
  \| \einterf_{km} \|^2 =\| (\mP - \mP_k) \vshat_m\|^2 = \vshat_m^\top \mP \vshat_m - \vshat_m^\top \mP_k \vshat_m,
  \end{align}
  and \eref{sir_base} follows by $d_m = \vshat_m^\top \mP \vshat_m$.
  Finally, it can be seen that $\starget_{km} + \einterf_{km} = \mP \vshat_k$ and thus \eref{sar_base} follows similarly.
\end{proof}

We can make a few observations.
Once the filters $\vh_{km}$ and $\vg_{km}$ have been computed, $c_{km}$ and $d_m$ only require an extra $O(KML)$ operations.
The SAR does not depend on the reference index $k$.
We call the cosine metrics thus because they are the squared cosine of angles $\alpha_{km}$ and $\gamma_m$ in \ffref{bsseval_illustration}.
Moreover, $\nicefrac{c_{km}}{d_m}$ is the square cosine of $\beta_{km}$.
This sheds light on the \SISDR that we now understand to be derived from the angle between the estimate and reference.

\subsection{Efficient Toeplitz Linear System Solver}
\seclabel{CGD}

We have established via Theorem~\ref{thm:bss_eval_from_base_metrics} that efficiently solving the Toeplitz and block-Toeplitz systems \eref{sdr_linear_system} and \eref{sir_linear_system} is key to the computations of the \bsseval metrics.
Direct solution by Gaussian elimination has cubic time in the matrix size. 
However, highly efficient solvers are typically available in numerical linear algebra libraries such as BLAS and Lapack.
The celebrated Levinson-Durbin recursion~\cite{levinsonWienerRootMean1946} works in quadratic time, but is seldom available in libraries.
There exists also an even better alternative.


The CGD algorithm with a circulant preconditioner has complexity $O(L \log L)$ for an $L \times L$ Toeplitz system~\cite{chanConjugateGradientMethods1996}.
We briefly review the method here applied to solving \eref{sdr_linear_system}.
CGD only requires matrix-vector multiplication by the system matrix $\mR_k$.
For a Toeplitz matrix, such as $\mR_k$, this can be done in $O(L \log L)$ operations by leveraging the FFT.
Convergence of CGD is dictated by the distribution of eigenvalues of $\mR_k$, and can be improved by a preconditioner.
For example, the optimal circulant matrix $\mC_k$ minimizing $\| \mC_k - \mR_k \|_F^2$~\cite{chanOptimalCirculantPreconditioner1988}.
For symmetric $\mR_k$ with first column $\vr = [r_1,\ldots,r_{L}]^\top$, the first column of $\mC_k$ is given by
$(\mC_k)_{11} = r_1$,
\begin{align}
  (\mC_k)_{\ell 1} =  L^{-1} \left[ (L - \ell + 1) r_\ell + (\ell - 1) r_{L - \ell + 1} \right], \ \ell \geq 2.
  \elabel{opt_circ}
\end{align}
It has been shown that the eigenvalues of $\mC_k^{-1} \mR_k$ cluster around 1, and only a few iterations are required until convergence, independent of the matrix size.
Multiplication by $\mC_k^{-1}$ is done in $O(L \log L)$ time via the FFT.
For $K$ systems, the cost is thus $O(K L \log L)$.

For the block-Toeplitz system \eref{sir_linear_system}, we construct the preconditioner by replacing the Toeplitz blocks of $\mR$ by their optimal circulant approximation.
The formula is slightly different than~\eref{opt_circ} because the off-diagonal blocks are not symmetric.
It can be found in~\cite{chanOptimalCirculantPreconditioner1988}.
By applying the FFT, we obtain a block-diagonal matrix with $L$ $K \times K$ blocks that we invert with a direct solver.
This is a one time cost of $O(L K^3)$.
Matrix-vector multiplication by $\mR$ or the preconditioner requires $O(K^2 L \log L)$ operations using the FFT.
Thus, solving \eref{sir_linear_system} requires $O(K^3 L + K^2 L \log L)$, a substantial saving of $O(L^2)$ compared to the direct method.
We also found in practice that the solution of \eref{sdr_linear_system} provides a good initial value to solve \eref{sir_linear_system}.

\section{Experiments}

We assess the proposed implementation with respect to other publicly available Python implementations.
\textbf{\mireval}\footnote{\url{https://github.com/craffel/mir_eval}}~\cite{raffelMirEvalTransparent2014} is the most widely used implementation and is regression tested to ensure the same output as the original Matlab implementation\footnote{\url{http://bass-db.gforge.inria.fr/bss_eval/}}.
\textbf{\sigsep}\footnote{\url{https://github.com/sigsep/bsseval}} is a recent re-implementation of the \mireval implementation with a focus on performance.
\textbf{\cisdr}\footnote{\url{https://github.com/fgnt/ci_sdr}} implements the \SDR only, but in a differentiable way, to be used to train neural networks.
Experiments where conducted on a Linux workstation with an Intel\textregistered\ Xeon\textregistered\ Gold 6230 CPU \@ \SI{2.10}{\giga\hertz} with 8 cores, an NVIDIA\textregistered\ Tesla\textregistered\ V100 graphical processing unit (GPU), and \SI{64}{\giga\byte} of RAM.
Throughout, we use \solve and \cgd to denote Gaussian elimination and 10 iterations of CGD, respectively.
We use \fpsingle and \fpdouble to mean single and double precision floating-point modes, respectively.

\textbf{Dataset} We use the dataset of reverberant noisy speech mixtures introduced in~\cite{scheiblerSurrogateSourceModel2021}.
The relative SNR of sources is selected at random in the range \SIrange{-5}{5}{\decibel}.
Speech and noise samples were selected from the WSJ1~\cite{wsj1} and CHIME3 datasets~\cite{barkerThirdCHiMESpeech2015}, respectively.
Noise is scaled to obtain a final SNR between \SIrange{10}{30}{\decibel}.
Mixtures contain two, three, and four sources, with an equal number of microphones.
For each number of sources, the dataset is split into training, validation, and test with \num{37416}, \num{503}, and \num{333} mixtures, respectively.

\begin{figure}[t!]
  \centering
  \includegraphics[width=\linewidth]{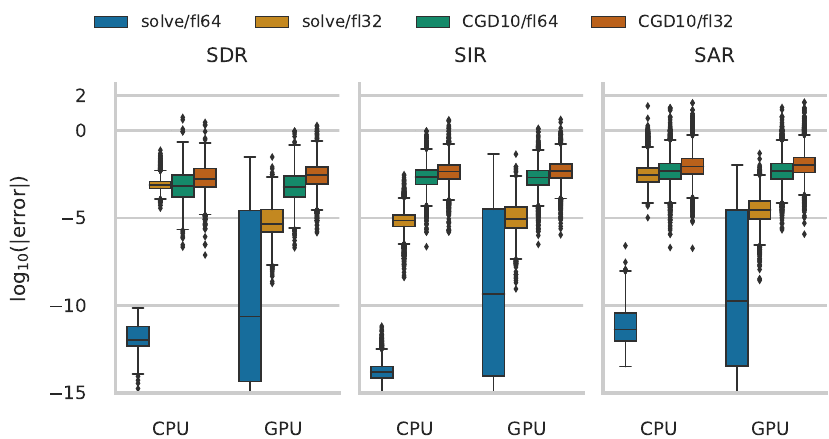}
  \caption{Box-plots of the log-absolute error relative to \mireval output.}
  \flabel{error_vs_mireval}
\end{figure}

\subsection{Evaluation on Speech Mixtures Dataset}
\seclabel{experiment_vs_mireval}

Our first experiment compares our proposed implementation to \mireval and \sigsep for the computation of \SDR, \SIR, and \SAR.
We use the test set, augmented by the output of separation by AuxIVA~\cite{onoStableFastUpdate2011}, doubling the number of samples.
We measure the difference with respect to \mireval's output and the runtime.

\ffref{error_vs_mireval} shows box-plots of the absolute difference with the output of \mireval.
Using \solve with \fpdouble shows very little difference with \mireval, less than \SI{e-6}{\decibel} in all cases on the CPU.
On the GPU, there is more variance, but, essentially, the error in decibels is negligible.
Switching to \fpsingle, the error is still small, between \SIrange{e-5}{e-3}{\decibel}, with the exception of \SAR on the GPU, where a few outliers exceed \SI{1}{\decibel}.
When the linear systems are solved approximately with \cgd, the median error is below \SI{e-2}{\decibel}.
There are some outliers with error above \SI{1}{\decibel}, more so for \SAR than other metrics.
We noted that the errors were zero-mean and thus do not impact averages over many samples.
Using \fpsingle or \fpdouble did not make a big difference when using CGD.

\ffref{runtime} shows runtimes averaged over the whole dataset.
We show results using CPU  with 1 core (1CPU), 8 cores (8CPU), and GPU.
In all cases the proposed implementation brings significant speed-up compared to \mireval and \sigsep.
There is about one order of magnitude speed up for two and three channels, and two orders for four channels.
About one and two orders of magnitude difference for three and four channels, respectively.
Going to the GPU brings yet another major speed-up.
\mireval and \sigsep seem to benefit less from multiple cores, which we attribute to \texttt{numpy} being less efficient in this area than \texttt{pytorch}.
For two channels using 8 CPU cores, \solve was faster than CGD.
In other cases, CGD is two to three times faster.
This advantage is more salient on GPU.

\begin{figure}[t!]
  \centering
  \includegraphics[width=\linewidth]{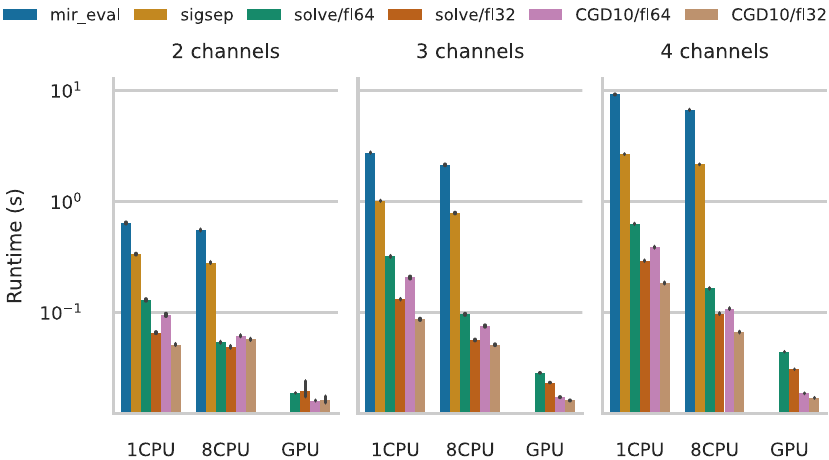}
  \caption{Average runtime (s) to compute \SDR, \SIR, and \SAR over the full test dataset. From left to right, 2, 3, and 4 channel signals. Within each subplot, we give runtimes for CPU with 1 core, 8 cores, and GPU.}
  \flabel{runtime}
\end{figure}

\subsection{Effects of Signal and Filter Length on Runtime}
\seclabel{experiment_runtime}

We analyze here the runtime behavior for \SDR only computation for varying signal and filter lengths, and number of channels.
We compare to the \cisdr toolbox~\cite{boeddeker_convolutive_2021} which provides an \SDR only implementation also based on \texttt{pytorch}.
The measurement are done by computing the \SDR for random signals on the GPU.
\tref{runtime_sdr_only_results} shows the average runtime of 10 measurements, each for a batch computation with 10 signals.
When using \solve, while the difference with \cisdr is small for two channels, the gap steadily widens with the number of channels.
The speed-up is also larger for longer signals, which is expected since the proposed implementation reduces computations involving the full length to a minimum.
Next, we confirm that using CGD leads to a large speed gain.
Most notably, doubling the filter size has no visible effect on the runtime for CGD, whereas it quadruples when using \solve.
The most extreme speed-ups occur for short signals and long filters, e.g., $10\times$ to $27\times$ depending on the number of channels.
This may be very valuable for methods where PIT is applied to many short signals, such as in uPIT~\cite{kolbaekMultitalkerSpeechSeparation2017}.

\begin{table}
  \caption{Runtime in \si{\milli\second} (speed-up) of the proposed method compared to the \cisdr~\cite{boeddeker_convolutive_2021} implementation. Runtimes are for batches of 10 signals on GPU.}
  \footnotesize
  \begin{tabular*}{\linewidth}{@{}l@{\extracolsep{\fill}}lr@{~}lr@{~}lr@{~}lr@{~}l@{}}
    \toprule
    \multicolumn{2}{@{}l@{}}{Filter length}  & \multicolumn{4}{c}{512 taps} & \multicolumn{4}{c}{1024 taps} \\
    \midrule
    \multicolumn{2}{@{}l@{}}{Signal length}  & \multicolumn{2}{c}{\SI{5}{\second}} & \multicolumn{2}{c}{\SI{20}{\second}} & \multicolumn{2}{c}{\SI{5}{\second}} & \multicolumn{2}{c}{\SI{20}{\second}} \\
    \midrule
    2 ch. & \cisdr &         27 &  &         43 &  &         87 &   &        104 &  \\
          & \solve &         21 &  ($1\times$) &         26 &  ($2\times$) &         67 &   ($1\times$) &         71 &  ($1\times$) \\
          & \cgd &          9 &  ($3\times$) &         15 &  ($3\times$) &          9 &  ($10\times$) &         15 &  ($7\times$) \\
    3 ch. & \cisdr &         44 &  &         80 &  &        144 &   &        181 &  \\
          & \solve &         25 &  ($2\times$) &         38 &  ($2\times$) &         84 &   ($2\times$) &         97  & ($2\times$) \\
          & \cgd &         10 &  ($4\times$) &         22 &  ($4\times$) &         10 &  ($14\times$) &         23 &  ($8\times$) \\
    4 ch. & \cisdr &         72 &  &        136 &  &        233 &   &        297 &  \\
          & \solve &         29 &  ($2\times$) &         50 &  ($3\times$) &         92 &   ($3\times$) &        112 &  ($3\times$) \\
          & \cgd &         12 &  ($6\times$) &         33 &  ($4\times$) &         12 &  ($19\times$) &         33 &  ($9\times$) \\
    5 ch. & \cisdr &         97 &  &        196 &  &        330 &   &        427 &  \\
          & \solve &         34 &  ($3\times$) &         65 &  ($3\times$) &        101 &   ($3\times$) &        132 &  ($3\times$) \\
          & \cgd &         15 &  ($6\times$) &         46 &  ($4\times$) &         14 &  ($22\times$) &         46 &  ($9\times$) \\
    6 ch. & \cisdr &        129 &  &        268 &  &        447 &   &        590 &  \\
          & \solve &         41 &  ($3\times$) &         86 &  ($3\times$) &        119 &   ($4\times$) &        163 &  ($4\times$) \\
          & \cgd &         18 &  ($7\times$) &         62 &  ($4\times$) &         18 &  ($25\times$) &         62 &  ($9\times$) \\
    7 ch. & \cisdr &        168 &  &        386 &  &        584 &   &        778 &  \\
          & \solve &         48 &  ($3\times$) &        108 &  ($4\times$) &        134 &   ($4\times$) &        193 &  ($4\times$) \\
          & \cgd &         22 &  ($8\times$) &         82 &  ($5\times$) &         22 &  ($26\times$) &         82 &  ($9\times$) \\
    8 ch. & \cisdr &        208 &  &        462 &  &        741 &   &        987 &  \\
          & \solve &         54 &  ($4\times$) &        134 &  ($3\times$) &        148 &   ($5\times$) &        226 &  ($4\times$) \\
          & \cgd &         27 &  ($8\times$) &        106 &  ($4\times$) &         27 &  ($27\times$) &        106 &  ($9\times$) \\
      \bottomrule
  \end{tabular*}
\tlabel{runtime_sdr_only_results}
\end{table}

\subsection{Training Neural Networks}
\seclabel{experiment_dnn}

Finally, we demonstrate the suitability of the proposed implementation for the training of separation networks.
We train neural source models for blind source separation using AuxIVA as described in~\cite{scheiblerSurrogateSourceModel2021}.
We compare the following loss functions, \SISDR, as in~\cite{scheiblerSurrogateSourceModel2021}, \SDR with $L=512$/\solve, as in~\cite{boeddeker_convolutive_2021},
 $L=512$/CGD10, and $L=1024$/CGD10.
The dataset is the one described above.
We use the reverberant image sources as target signals.
Algorithm and training details are described in~\cite{scheiblerSurrogateSourceModel2021}, but the DNN used is the one from~\cite{boeddeker_convolutive_2021}.
We train for 57 epochs with initial learning rate of $3\times 10^{-4}$.
Training is done on two sources mixtures, test on two, three, and four sources mixtures.

\tref{dnn_training_results} shows the test results for the models with smallest validation loss.
We evaluate in terms of \SDR ($L=512$/\solve) and word error rate (WER) of an ASR model pre-trained using the \texttt{wsj/asr1} recipe of ESPNet~\cite{Watanabe:2018gy}.
First, we note that using the approximate CGD solver has no effect on the final accuracy.
The results are in fact remarkably similar.
Then, as in~\cite{boeddeker_convolutive_2021}, we observe that allowing longer distortion filters with $L>1$ leading to improvements of both $\SDR$ and WER.
In fact, we observe that $L=512$ might still be too short as $L=1024$ leads to better performance.


\begin{table}
  \caption{Mean \SDR (dB) / WER (\%) for blind source separation with a neural source model trained with the \SDR as loss using different parameters.
    For evaluation, the \SDR uses filter length $L=512$ taps and \solve.
  }
  \footnotesize
  \begin{tabular*}{\linewidth}{@{}l@{\extracolsep{\fill}}rccc@{}}
      \toprule
      Solver & $L$ & 2 ch. & 3 ch. & 4 ch. \\
      \midrule
      (\SISDR)              & 1    & 11.26 /  31.58 &  8.48 /  44.16 & 6.44 /  54.92 \\
      \solve  & 512  & 11.50 /  30.19 &  8.76 /  42.62 & 6.61 /  54.83 \\
      \cgd          & 512  & 11.50 /  30.18 &  8.76 /  42.65 & 6.61 /  54.88 \\
      \cgd          & 1024 & 11.60 /  29.42 &  8.95 /  41.95 & 6.92 /  51.47 \\
      \bottomrule
    \end{tabular*}
  \tlabel{dnn_training_results}
\end{table}


\section{Conclusion}

We introduced an improved algorithm to implement the widely used \bsseval metrics for blind source separation evaluation.
First, we reduce to a minimum computations that depend on the full length of input signals.
Second, we propose to use an iterative solver to find the optimal distortion filters.
We find very large runtime reductions that can potentially reduce the evaluation time from days to hours in BSS experiments.
The loss of accuracy due to the iterative solver does not impact average evaluation on datasets.
Furthermore, it opens the door to using longer distortion filters.
Experimental results suggest this may be beneficial to train separation networks.
It also makes it possible to use \bsseval on signals sampled at a higher frequency, e.g., \SI{44}{\kilo\hertz}.
We release our implementation as a Python package that can be used with both \texttt{numpy} and \texttt{pytorch}.

\clearpage

\section{References}

\printbibliography[heading=none]

\end{document}

%% file: project_illustrations.tex
\definecolor{myblue}{rgb}{0.00392156862745098, 0.45098039215686275, 0.6980392156862745}
\definecolor{mygreen}{rgb}{0.00784313725490196, 0.6196078431372549, 0.45098039215686275}
\definecolor{myred}{rgb}{0.8352941176470589, 0.3686274509803922, 0.0}

\tdplotsetmaincoords{60}{105}

\pgfmathsetmacro{\x}{3}
\pgfmathsetmacro{\y}{3.5}
\pgfmathsetmacro{\z}{3.5}
\pgfmathsetmacro{\d}{0.3}

\begin{tikzpicture}[tdplot_main_coords]
  \footnotesize

  \draw[thick,->] (0,0,0) -- (3.8,0,0) node[anchor=north west,xshift=-5mm]{$k$th reference subspace};
  \draw[thick,->] (0,0,0) -- (0,4.5,0) node[anchor=north,align=left]{subspace of \\ other references};
  \draw[thick,->] (0,0,0) -- (0,0,1.5) node[anchor=west]{noise subspace};

  \coordinate (ori) at (0,0,0);
  \coordinate (P) at (\x,\y,\z);
  \coordinate (Pxy) at (\x,\y,0);
  \coordinate (Px) at (\x,0,0);
  \coordinate (Py) at (0,\y,0);
  \coordinate (Pz) at (0,0,\z);

  \draw[-{Latex[length=2mm]}] (ori) -- (P);
  \draw[dashed] (ori) -- (Pxy);
  \draw[dashed] (Py) -- (Pxy);
  \draw[-{Latex[length=2mm]}] (ori) -- (Px);
  \draw[-{Latex[length=2mm]}] (Px) --node[above] {$\einterf_{km}$} (Pxy);
  \draw[-{Latex[length=2mm]}] (Pxy) -- node[right,yshift=6mm] {$\eartif_m$} (P);

  \node[right] at (P) {$\vshat_m$};
  \node[right] at (Pxy) {$\mP \vshat$};
  \node[left] at (Px) {$\starget_{km} = \mP_k \vshat_m$};

  \tdplotgetpolarcoords{\x}{\y}{\z};

  \tdplotdrawarc[thick,myblue]{(ori)}{0.9}{0}{\tdplotresphi}{below}{$\beta_{km}$}

  \tdplotsetthetaplanecoords{\tdplotresphi};
  \tdplotdrawarc[tdplot_rotated_coords,thick,myred]{(ori)}{0.9}{\tdplotrestheta}{90}{right,yshift=2mm}{$\gamma_{m}$};

  \tdplotdefinepoints(0,0,0)(\x,0,0)(\x,\y,\z)
  \tdplotdrawpolytopearc[thick,mygreen]{0.9}{anchor=east,xshift=-2mm}{\textcolor{mygreen}{$\alpha_{km}$}};


  \coordinate (Pxm1) at ($(\x,0,0)-(\d,0,0)$);
  \coordinate (Pxm3) at (\x,\d,0);
  \coordinate (Pxm2) at ($(\x,\d,0)-(\d,0,0)$);
  \draw (Pxm1) -- (Pxm2) -- (Pxm3);

  \coordinate (Pxym1) at ($(Pxy)-(0,\d,0)$);
  \coordinate (Pxym2) at ($(Pxy)+(0,-\d,\d)$);
  \coordinate (Pxym3) at ($(Pxy)+(0,0,\d)$);
  \draw (Pxym1) -- (Pxym2) -- (Pxym3);

  \coordinate (Pym1) at ($(Pxy)-(\d,0,0)$);
  \coordinate (Pym2) at ($(Pxy)+(-\d,0,\d)$);
  \coordinate (Pym3) at ($(Pxy)+(0,0,\d)$);
  \draw (Pym1) -- (Pym2) -- (Pym3);

\end{tikzpicture}